\documentclass[a4paper,11pt,reqno,twoside]{amsart}
\numberwithin{equation}{section}

\usepackage[latin1]{inputenc}
\usepackage[T1]{fontenc}

\usepackage{graphicx}
\usepackage{epstopdf}

\usepackage{color}
\definecolor{MyLinkColor}{rgb}{0,0,0.4}
\renewcommand{\bfseries}{\color[rgb]{0,0,0.4} \normalfont}

\newcommand{\N}{\mathbb{N}}
\newcommand{\OO}{\mathcal{O}}
\newcommand{\R}{\mathbb{R}}

\newcommand{\Z}{\mathbb{Z}}

\DeclareMathOperator{\arccot}{arccot}
\DeclareMathOperator{\arctanh}{arctanh}
\DeclareMathOperator{\Diff}{D}

\DeclareMathOperator{\sign}{sgn}

\newtheorem{thm}{Theorem}[section]

\newtheorem{prop}[thm]{Proposition}

\theoremstyle{definition}
\newtheorem{remark}[thm]{Remark}
\newtheorem{solution}{Wave class}

\title[Critical-layer water waves]{Steady water waves with multiple critical layers: interior dynamics}


\author[Ehrnstr\"om]{Mats Ehrnstr\"om}
\address{Institut f{\"u}r Angewandte Mathematik, Leibniz Universit{\"a}t Hannover, Welfengarten~1, 30167 Hannover, Germany. }
\email{ehrnstrom@ifam.uni-hannover.de}

\author[Escher]{Joachim Escher}
\address{Institut f{\"u}r Angewandte Mathematik, Leibniz Universit{\"a}t Hannover, Welfengarten~1, 30167 Hannover, Germany. }
\email{escher@ifam.uni-hannover.de}

\author[Villari]{Gabriele Villari}
\address{Dipartimento di Matematica, Viale Morgagni 67/A, 50134 Firenze, Italy.}
\email{villari@math.unifi.it}

\usepackage[colorlinks=true,linkcolor=MyLinkColor,urlcolor=blue,citecolor=red]{hyperref} 

\begin{document}

\begin{abstract}
We study small-amplitude steady water waves with multiple critical layers. Those are rotational two-dimensional gravity-waves propagating over a perfect fluid of finite depth. It is found that arbitrarily many critical layers with cat's-eye vortices are possible, with different structure at different levels within the fluid. The corresponding vorticity depends linearly on the stream function.  
\end{abstract}

\maketitle

\section{Introduction}
This paper is concerned with steady, periodic, two-dimensional gravity-waves of permanent shape and velocity. Famous among these are the \emph{Stokes waves}; symmetric waves with a surface profile which rises and falls once in every minimal period \cite{MR1317348, MR1629555, MR642980, Stokes47, MR1422004}. Of particular interest are waves with vorticity. Vorticity plays a crucial role for wave-current interactions and in the formation of wind-generated waves \cite{Craik1988,Saffman1981, ThomKlop97}. The first mathematical construction of a rotational free-surface fluid flow is due to Gerstner in the beginning of 19th century \cite{MR1819940,MR2434727}, but it was first with the pioneering investigation \cite{ConstantinStrauss04} that a modern theory of both large and small rotational waves was established. This theory was extended to deep water waves in \cite{MR2215274}, to   small-amplitude waves with surface tension in \cite{MR2262949}, and to large-amplitude waves with surface tension and stratification in \cite{W09}.

Many properties inherent in irrotational periodic gravity waves, such as the symmetry of the surface profile \cite{MR1869386,MR1734884}, the analyticity of the streamlines \cite{EC2010}, and the Stokes conjecture \cite{MR666110, Plotnikov1982}, carry over to rotational water waves \cite{MR2362244,EC2010,MR2144685, MR2256915, Varvaruca20094043}. A notable exception from this rule arises when one examines flows with internal stagnation, i.e. points where the velocity of a fluid particle coincides with that of the wave itself. Even when the vorticity is only constant, \emph{critical layers} with cat's-eye vortices arise \cite{MR2409513}. Those are horizontal layers of closed streamlines separating the fluid into two disjoint regions, a behaviour that is not possible for irrotational waves~\cite{MR2257390}. Recently, the existence of such waves with one critical layer as solutions of the full Euler equations was established \cite{Wahlen09} (see also~\cite{cv09} and, for a study of stagnation points in rotational flows, \cite{MR2329144}). This constituted an important connection between the mathematical research on exact Stokes waves and the study of waves with critical layers within the wider field of fluid dynamics (cf. \cite{Choi2003, MR851672, MR668167,Saffman1995, Thorpe1981}). 

In the very recent investigation~\cite{eew10} the theory from~\cite{Wahlen09} was extended to the case of an affine vorticity distribution, yielding the existence of exact small-amplitude gravity waves with \emph{arbitrary many critical layers}. It is our aim to give a qualitative, and to some extent also a quantitative, description of those waves (\cite{eew10} also contains a proof for the existence of exact bichromatic waves with vorticity, but those are not considered here). This is a natural continuation of a line of recent research on the qualitative features of various free-surface flows (see, e.g., \cite{EV09}, and the references \cite{MR2422547, MR2318158, MR2604871, MR2383410, MR2272104, MR2287829, MR2439570, MR2434722, IonescuKruse2009} from that survey). In particular, the results here obtained could be used to give a description of the particle trajectories within linear waves with critical layers.      

The focus of this paper is, however, slightly different. Since the construction in \cite{eew10} is by bifurcation from laminar flows, for small waves it is possible to investigate the flow and give exact estimates on the error (see Proposition~\ref{prop:exact}). It turns out that waves with an affine vorticity distribution can be naturally divided into four wave classes, and we give the velocity fields and the bifurcation relations in the different cases. Theorem~\ref{thm:main} provides the qualitative description of the wave class with multiple critical layers. Although there is a rich variety of flow configurations we discern two main scenarios (which, essentially, can be seen in Figure~\ref{fig:wave24}). An interesting feature is that, apart form the region closest to the free surface, the fluid motion takes place in vertically disconnected, \emph{completely flat} regions (which is not the case in waves with a single critical layer). In Theorem~\ref{thm:stagnation} we are also able to give a \emph{quantitative} description of the levels at which stagnation points, and therefore critical layers, can arise. Those results are graphically captured in Figure~\ref{fig:zeros}.

The disposition is as follows. Section~\ref{sec:preliminiaries} describes the governing equations, with~Section~\ref{sec:laminar} narrowing in on laminar flows and the first-order perturbations thereof. In Section~\ref{sec:classes} we describe the four wave classes, and detail at which levels stagnation can occur in the different cases. Finally, Section~\ref{sec:portraits} presents the main structure of the interesting wave class with multiple critical layers, and some numerical examples are given. For a quick glance at the waves, see the last section.

\section{Preliminaries}\label{sec:preliminiaries}
Let $(x,y)$ be Cartesian position coordinates, and  $(u,v) = (\dot x, \dot y)$ the corresponding velocity field. Here
\[
u := u(t,x,y), \qquad v:= v(t,x,y),
\]
are $2\pi$-periodic in the $x$-variable and the vertical coordinate $y$ ranges from the flat bed at $y = 0$ to the (normalized) free water surface at $y = 1 + \eta(t,x)$. Let $p := p(t,x,y)$ denote the \emph{pressure}, and $g$  the \emph{gravitational constant of acceleration}. In the mathematical theory of steady waves it is common and physically realistic to consider water as inviscid and of constant density \cite{MR1629555,MR642980}. The Euler equations
\begin{subequations}\label{eq:timeproblem}
\begin{equation}\label{eq:timeeuler}
\begin{aligned}
u_t + u u_x + v u_y &= - p_x,\\
v_t + u v_x + v v_y &= -p_y -g,
\end{aligned}
\end{equation}
then model the motion within the fluid. The equations
\begin{equation}\label{eq:timecurl}
u_x + v_y = 0 \qquad\text{ and }\qquad v_x - u_y = \omega
\end{equation}
additionally describes incompressibility and the \emph{vorticity} $\omega$, respectively.\footnote{This sign convention for the vorticity is consistent with~\cite{eew10}. The reader should be advised that the vorticity may also appear with the opposite sign in the literature.} At the surface the conditions
\begin{equation}\label{eq:timesurface}
p = p_0 \quad\text{ and }\quad v = \eta_t + u \eta_x
\end{equation}
separate the air from the water, $p_0$ being the \emph{atmospheric pressure}. Note that the second condition in~\eqref{eq:timesurface} states that $y(t) - \eta(t,x(t))$ is constant over time, so that the same particles constitute the interface at all times. Similarly, no water penetrates the flat bed, whence we have
\begin{equation}\label{eq:timebed}
v = 0 \qquad\text{ at }\quad y = 0.
\end{equation}
\end{subequations}
The equations \eqref{eq:timeproblem} govern the motion of two-dimensional gravitational water waves on finite depth. 

An important class of waves are \emph{travelling waves}, propagating with constant shape and speed. Mathematically, such waves are solutions of \eqref{eq:timeproblem}  with an $(x-ct)$-dependence, where $c > 0$ is the constant wavespeed, and we have restricted attention to waves travelling rightward with respect to the fixed Cartesian frame. Since $\Diff_t (x-ct) = u - c$, it is natural to introduce \emph{steady variables},
\[
X := x - ct, \qquad U := u-c. 
\]
We shall also write $Y$ for $y$ and $V$ for $v$ to indicate when we are in the travelling frame. In the steady variables the fluid occupies
\begin{equation}\label{eq:omegaeta}
\Omega_\eta := \left\{ (X,Y) \in \R^2 \colon 0 < Y < 1 + \eta(X) \right\}.
\end{equation}
Define the \emph{relative pressure} $P$ through
\[
p =: p_0 + g(1+P-Y).
\]
Since the term $-gY$ measures the hydrostatic pressure distribution, the relative pressure is a measure of the pressure perturbation induced by a passing wave. Altogether we obtain the governing equations
\begin{subequations}\label{eqs:steadyproblem}
\begin{equation}\label{eq:steadyeuler}
\begin{aligned}
U U_X + V U_Y &= - g P_X,\\
U V_X + V V_Y &= -g P_Y,\\
U_X + V_Y &= 0,\\
V_X - U_Y &= \omega,
\end{aligned}
\qquad\quad\text{ in }\quad \Omega_\eta
\end{equation}
with boundary conditions
\begin{equation}\label{eq:steadysurface}
\begin{aligned}
P &= \eta,\\ 
V &= U \eta_X,
\end{aligned}
\qquad\quad \text{ on }\quad Y = 1 + \eta(X),
\end{equation}
and
\begin{equation}\label{eq:steadybed}
V = 0, \qquad\quad\text{ on }\quad Y = 0.
\end{equation}
\end{subequations} 
The problem of finding $(U, V, P, \eta)$ such that \eqref{eqs:steadyproblem} is satisfied is known as the \emph{steady water-wave problem}. Since $\eta$ is an \emph{a priori} unknown, \eqref{eqs:steadyproblem} is a free-boundary problem. 

\subsection*{The $\alpha$-problem}\label{subsec:alpha}
When $\eta \in C^3(\R)$, and $u, v \in C^2(\overline \Omega_\eta)$, one can use the fact that the velocity field is divergence-free (cf. \eqref{eq:timecurl}) to introduce a \emph{stream function} $\psi \in C^3(\overline\Omega)$ with
\begin{equation}\label{eq:psi}
\psi_X := -V \quad\text{ and }\quad  \psi_Y := U.
\end{equation}
Define the Poisson bracket $\{ f, g \} := f_X  g_Y - f_Y g_X$.
\begin{prop}[Stream-function formulation]
The water-wave problem \eqref{eqs:steadyproblem} is equivalent to that 
\begin{equation}\label{eq:streamproblem}
\begin{aligned}
\Delta \psi &= - \omega, \qquad&\text{ in }\qquad &\Omega_\eta,\\
\{ \psi, \Delta \psi \} &= 0, \qquad&\text{ in }\qquad &\Omega_\eta,\\
\left| \nabla \psi \right|^2 + 2gy &= C, \qquad&\text{ on }\qquad &Y = 1 + \eta(X),\\ 
\psi &= m_1, \qquad&\text{ on } \qquad &Y = 1 + \eta(X),\\ 
\psi &= m_0, \qquad&\text{ on } \qquad &Y=0,
\end{aligned}
\end{equation}
for some constants $m_0$, $m_1$, and $C$.
\end{prop}

\begin{proof}
Identify $\psi$ with $U$ and $V$ through \eqref{eq:psi}. Given the regularity assumptions and that $\Omega_\eta$ is simply connected, we see that $U_X + V_Y = 0$ is equivalent to the existence of $\psi$. The relations $P = \eta$ in \eqref{eq:steadysurface} and $V=0$ in \eqref{eq:steadybed} mean that $\psi$ is constant on the surface and on the flat bed, just as $V_X - U_Y = -\omega$ means that $\Delta \psi = -\omega$. It remains to show how the equations of motion relate to the Bernoulli surface condition and the bracket condition. 

Given \eqref{eqs:steadyproblem}  one can eliminate the relative pressure by taking the curl of the Euler equations. That yields
\begin{subequations}\label{eqs:nopressure}
\begin{equation}\label{eq:curledeuler}
U \Delta V - V \Delta U = 0.
\end{equation}
Moreover, by differentiating the relation $P = \eta$ along the surface, and using \eqref{eq:steadyeuler}, we find that
\begin{equation}\label{eq:bernoullisurface}
U^2 + V^2 + 2gY = C, \qquad Y = 1 + \eta(X).
\end{equation}
\end{subequations}
Hence \eqref{eq:streamproblem} holds. Contrariwise, if $(U,V)$ fulfil \eqref{eq:curledeuler} and \eqref{eq:bernoullisurface}, one can define $P$ up to a constant through \eqref{eq:steadyeuler}, and, with the right choice of constant, $P$ satisfies \eqref{eq:steadysurface}. 
\end{proof}


Consider now the case when $\psi_Y$ may vanish, but ${\Delta \psi_Y} / \psi_Y$ can be extended to a continuous function on $\overline\Omega_\eta$, i.e. 
\begin{equation}\label{eq:alpha}
\alpha := \frac{\Delta \psi_Y}{\psi_Y} \in C^0(\overline\Omega_\eta)
\end{equation}
One can then exchange the bracket condition $\{ \psi, \Delta \psi \} = 0$ in~\eqref{eq:streamproblem} for 
\begin{equation}\label{eq:streameuleralpha}
\left( \Delta - \alpha \right) \nabla \psi = 0.  
\end{equation}
When $\alpha$ is a constant there exists an affine \emph{vorticity function} $\gamma$ with $\gamma^\prime = -\alpha$, meaning that
\[
\Delta \psi = -\gamma(\psi) = \alpha \psi + \beta, \qquad \beta \in \R.
\]
Observe that this does not rule out the existence of \emph{stagnation points} $\nabla \psi = 0$. Without loss of generality we may take $\beta$ to be zero; changing it corresponds to changing $m_0$ and $m_1$. The choice $\alpha = 0$ models constant vorticity and was investigated in \cite{MR2409513, Wahlen09}. The next natural step is a constant but non-vanishing $\alpha$. That is the setting of this investigation.

\section{Laminar flows and their first-order perturbations}\label{sec:laminar}
\emph{Laminar flows} are solutions of the steady water-wave problem \eqref{eqs:steadyproblem} with $\eta(X) =  0$. Those are the \emph{running streams} for which 
\[
U(X,Y) = U_0(Y) \quad \text{ and }\quad V = P = \eta = 0.
\]
We shall require that $U_0 \in C^2([0,1], \R)$. The function $U_0$ is the (rotational) \emph{background current}, upon which we will impose a small disturbance: the system \eqref{eqs:steadyproblem} will be linearized at a point $(U,V,P,\eta) = (U_0,0,0,0)$, and the solutions of the constructed linear problem analyzed. We thus assume that $U$, $V$, $P$ and $\eta$ allow for expansions of the form
\begin{equation}\label{eq:expansion}
f =  f_0 + \varepsilon f_1 + {\mathcal O}\left(\varepsilon^2\right), \quad\text{ as }\quad \varepsilon \to 0. 
\end{equation}
Here $U_0$ is a background current as described above, and
\[
V_0 = P_0 = \eta_0 = 0.
\]
By inserting these expansions into \eqref{eqs:steadyproblem}, and retaining only first-order terms in $\varepsilon$, we obtain the linearized system  
\begin{subequations}\label{eqs:final}
\begin{equation}\label{eq:finaleuler}
\begin{aligned}
    \partial_X {U_1} + \partial_Y V_{1} &= 0,\\
    U_0 \, \partial_X U_1 + V_1\, \partial_Y U_{0} &= - \partial_X P_{1},\\
    U_0 \, \partial_X V_1   &= - \partial_Y P_{1},
\end{aligned}
\qquad\text{ in }\quad \R \times (0,1)
\end{equation}
with boundary conditions 
\begin{equation}\label{eq:finalsurface}
\begin{aligned}
V_1 &= U_0 \, \partial_X \eta_{1},\\ 
P_1 &= \eta_1, 
\end{aligned}
\qquad\text{ on }\quad Y = 1,
\end{equation}
as well as
\begin{equation}
V_1 = 0 \qquad \text{ on }\quad Y = 0.
\end{equation}
\end{subequations}
The following result allows us to eliminate the relative pressure from~\eqref{eqs:final}.

\begin{prop}\label{prop:vsystem}
Let the background current $U_0$ be given. Under the condition that 
\begin{equation}\label{eq:normalization}
\int_{-\pi}^{\pi} \eta_1(X)\,dX = 0 \quad\text{ and }\quad \int_{-\pi}^{\pi} U_1(X,Y)\,dX = 0, \quad Y \in [0,1],
\end{equation}
the solutions $(U_1,V_1,P_1,\eta_1)$ of \eqref{eqs:final} are in one-to-one-correspondence with the solutions $V_1$ of
\begin{equation}\label{eq:vsystem}
\begin{aligned}
U_0 \Delta V_1 &= U_0^{\prime\prime} V_1, \qquad &0<Y < 1,\\
(1+U_0 U_0^\prime) V_1 &= U_0^2 \, \partial_Y V_1, \qquad	& Y =1,\\
V_1 &= 0, \qquad &Y=0.
\end{aligned}  
\end{equation}
\end{prop}

\begin{proof}
Taking the curl of the linearized Euler equations, and
differentiating $p = \eta$ along the
linearized surface $Y=1$ yields~\eqref{eq:vsystem}.
If $(U_1,V_1,P_1,\eta_1)$ is a solution of \eqref{eqs:final}
then $V_1$ fulfills \eqref{eq:vsystem},
and if $V_1$ is a solution of \eqref{eq:vsystem},
then one can find $(U_1,P_1,\eta_1)$ such that
\eqref{eqs:final} holds. One defines $U_1$  through the first equation in \eqref{eq:finaleuler},
and then $P_1$ through the two last equations in \eqref{eq:finaleuler}. The linear surface $\eta_1$ can 
be determined by \eqref{eq:finalsurface}, and the boundary condition at $Y=1$ in \eqref{eq:vsystem}
guarantees that \eqref{eq:finalsurface} is consistent with \eqref{eq:finaleuler}. Notice, however, that for a given $V_1$, a solution $U_1$ is only determined modulo functions $f(y)$, and $\eta_1$ up to a constant. We shall therefore require that the periodic mean of the first-order solution equals that of the running stream, meaning that \eqref{eq:normalization} holds. In particular, this implies that the solution $(U_1, V_1, P_1, \eta_1)$ of \eqref{eqs:final} is unique with respect to the solution $V_1$ of \eqref{eq:vsystem}. 
\end{proof}
  
\subsection*{Laminar vorticity}  
Now, suppose that $U_0^{\prime\prime}/U_0$ can be extended to a continuous function on $[0,1]$, and introduce the \emph{laminar vorticity}
\begin{equation}\label{eq:alpha_0}
\alpha_0 := \frac{U_0^{\prime\prime}}{U_0} \in C([0,1],\R).
\end{equation}
Let (cf. \eqref{eq:vsystem})
\[
\mu_1 := 1+U_0(1) U_0^\prime(1) \quad\text{ and }\quad \mu_2 := U_0^2(1).
\]
According to Proposition~\ref{prop:vsystem} we may then consider the system
\begin{equation}\label{eq:alphalinear}
\begin{aligned}
\Delta V_1 &= \alpha_0 V_1, \qquad &0< Y < 1,\\
\mu_1 V_1 &= \mu_2 \, \partial_Y V_1, \qquad	&Y =1,\\
V_1 &= 0, \qquad &Y=0.
\end{aligned}
\end{equation}  
In our case $\alpha = \alpha_0 \in \R$, and constant vorticity is captured by $\alpha_0 = 0$.
  
\subsection*{Relation to exact nonlinear solutions}
In what comes we will find and investigate four solution classes of \eqref{eq:alphalinear} and thus of \eqref{eqs:final}. 
Any exact solution of the steady water-wave problem~\eqref{eq:streamproblem} with $\Delta \psi = \alpha \psi$ that allows for an expansion as in \eqref{eq:expansion} and adheres to the normalization \eqref{eq:normalization} will satisfy the velocity fields here investigated up to an error of order $\varepsilon^2$ in the appropriate space. A particularly case (Wave class~\ref{class:1} on page~\pageref{class:1}) corresponds to a class of solutions found in \cite{eew10} by linearizing around a running stream with background current $U_0(Y) = a \sin(\theta_0 (Y - 1)+ \lambda)$. Here $\theta_0 = \sqrt{|\alpha_0|}$ and $\lambda \in \R$ is a parameter. Those solutions do not necessarily satisfy the normalization \eqref{eq:normalization}; while $\int_{-\pi}^\pi \eta_1\,dX=0$ the strength of the first-order background current may change with $\varepsilon$. This is the reason why $a$ depends on $\varepsilon$ in the following proposition, which is a consequence of the results from \cite{eew10}. In accordance with~\eqref{eq:omegaeta} we let, for any small and positive constant $\delta$, the set $\Omega_{-\delta}$ denote the part of the fluid domain where $0 < Y < 1 - \delta$.

\begin{prop}\label{prop:exact}
Let $\varepsilon \mapsto (\psi,\eta) \in C^2(\overline\Omega_\eta) \times C^2(\R)$ be a solution curve found in~\cite{eew10} by bifurcation from a one-dimensional kernel of minimal period $2\pi$. Pick $0 < \delta << 1$. For any $\varepsilon$ small enough, there exists $a = a(\varepsilon)$ such that the velocity field $(U,V) = (\psi_Y,-\psi_X)$ coincides with that of wave class~\ref{class:1} with $U_0(Y) = a(\varepsilon) \sin(\theta_0 (Y - 1)+ \lambda)$ up to addition of terms ${\mathcal O}(\varepsilon^2)$ in $C^2(\overline\Omega_{-\delta})$. The map $\varepsilon \mapsto a(\varepsilon)$ is smooth and $a(0)$ fulfils the bifurcation condition \eqref{eq:C1C2<-1}.
\end{prop}

\section{Wave classes}\label{sec:classes}
Even when $\alpha_0$ is a constant, the linear system \eqref{eq:alphalinear} contains a rich variety of solutions, including asymmetric ones (cf. \cite{07022009}). We shall see that restricting attention to the first Fourier mode of $V_1$ still produces a wide range of linear waves. We thus search for a solution of the form     
\begin{equation}\label{eq:ansatz}
V_1 = \sin(X) f(Y), \qquad f \in C^2([0,1],\R). 
\end{equation}
The \emph{ansatz} \eqref{eq:ansatz} reduces the system \eqref{eq:alphalinear} to a (trivial) Sturm--Liouville problem: 
\begin{equation}
\begin{aligned}\label{eq:sturm}
-f^{\prime\prime} + (\alpha_0 +1) f &= 0,\\
\mu_1 f(1) - \mu_2 f^\prime(1) &=0,\\
f(0) &= 0,
\end{aligned}
\end{equation}
with $\alpha_0 \in \R$,  $\mu_1^2 + \mu_2^2 > 0$, and $\mu_2 \geq 0$.
Since the case $\alpha_0 = 0$ has already been treated in \cite{MR2409513} we restrict our attention to $\alpha_0 \neq 0$, corresponding to non-constant vorticity. Define
\[
\theta_0 := \sqrt{|\alpha_0|} \quad\text{ and }\quad \theta_1:= \sqrt{|\alpha_0 +1|}.
\]
Using the Sturm--Liouville problem~\eqref{eq:sturm} to determine $V$, and then $U$ via Proposition~\ref{prop:vsystem}, we find that the solutions belong to one of the following four classes:
\medskip

\begin{solution}[Laminar vorticity $\alpha_0 < -1$]\label{class:1}
The solutions of \eqref{eq:sturm} are generated by $f(Y) = \sin(\theta_1 Y )$ with
\begin{align}
U_0(Y) &= a \sin(\theta_0 (Y - 1)+ \lambda),\notag\\
a^{-2} &= \sin^2{(\lambda)} \left( \theta_1 \cot(\theta_1)  - \theta_0 \cot(\lambda) \right),\label{eq:C1C2<-1}\\
\lambda &\in \left(\arccot\left({\textstyle\frac{\theta_1 \cot(\theta_1)}{\theta_0}}\right),\pi\right).\notag
\end{align}
Up to the first order in $\varepsilon$, 
\begin{equation}\label{eq:alpha<-1}
\begin{aligned}
U(X,Y) & = U_0(Y) + \varepsilon \theta_1 \cos{X} \cos(\theta_1 Y),\\
V(X,Y) &= \varepsilon \sin{X} \sin(\theta_1 Y).\\
\end{aligned}
\end{equation}  

\end{solution}

\medskip
    
\begin{solution}[Laminar vorticity $\alpha_0 = -1$]\label{class:2}
The solutions of~\eqref{eq:sturm} are generated by $f(Y) = Y$ with 
\begin{align}
U_0(Y) &= a \sin(Y - 1+ \lambda),\notag\\
a^{-2} &=  \sin^2(\lambda) \left( 1 - \cot(\lambda) \right), \qquad \lambda \in \left({\textstyle\frac{\pi}{4}},\pi\right). \label{eq:C1C2-1}
\end{align}
Up to the first order in $\varepsilon$, 
\begin{equation}\label{eq:alpha-1}
\begin{aligned}
U(X,Y) & = U_0(Y) + \varepsilon \cos(X),\\
V(X,Y) &= \varepsilon Y \sin(X).
\end{aligned}
\end{equation}  
\end{solution}
\medskip

\begin{solution}[Laminar vorticity $-1 < \alpha_0 < 0$]\label{class:3}
The solutions of \eqref{eq:sturm} are generated by $f(Y) = \sinh(\theta_1 Y)$ with 
\begin{align}
U_0(Y) &= a \sin(\theta_0 (Y - 1)+ \lambda),\notag\\
a^{-2} &= \sin^2{(\lambda)} \left( \theta_1 \coth(\theta_1)  - \theta_0 \cot(\lambda) \right),\label{eq:C1C2>-1<0}\\
\lambda &\in \left(\arccot\left({\textstyle\frac{\theta_1 \coth(\theta_1)}{\theta_0}}\right),\pi\right).\notag
\end{align}
Up to the first order in $\varepsilon$, 
\begin{equation}\label{eq:alpha>-1<0}
\begin{aligned}
U(X,Y) & = U_0(Y) + \varepsilon \theta_1 \cos{X} \cosh(\theta_1 Y),\\
V(X,Y) &= \varepsilon \sin{X} \sinh(\theta_1 Y).
\end{aligned}
\end{equation}  
\end{solution}
\medskip
 
\begin{solution}[Laminar vorticity $\alpha_0 > 0$]\label{class:4}
The solutions of~\eqref{eq:sturm} are generated by $f(Y) = \sinh(\theta_1 Y )$ with
\begin{align}
U_0(Y) &= a \sinh(\theta_0 (Y - 1)) + \lambda \cosh(\theta_0 (Y-1)),\notag\\
a &= \frac{ \lambda^2  \theta_1 \coth(\theta_1) - 1}{ \lambda \theta_0 }, \qquad \lambda \neq 0.\label{eq:C1C2>0}
\end{align}
Up to the first order in $\varepsilon$, 
\begin{equation}\label{eq:alpha>0}
\begin{aligned}
U(X,Y) &= U_0(Y)  + \varepsilon \theta_1 \cos{X} \cosh(\theta_1 Y),\\
V(X,Y) &= \varepsilon \sin{X} \sinh(\theta_1 Y).
\end{aligned}
\end{equation}  
\end{solution}
\medskip

\subsection*{Stagnation}

The explicit solutions allow us to determine the possible levels of stagnation. From the following result one obtains Figure~\ref{fig:zeros}. 

\begin{thm}[Stagnation]\label{thm:stagnation}
The following hold for the background current $U_0$ for the wave classes \ref{class:1}--\ref{class:4}:
\begin{itemize}
\item[W~\ref{class:1}.] For any $Y_0 \in [0,1)$ there exists $\alpha_0 < -1$ and $\lambda$ such that $U_0(Y_0) = 0$, and the number of zeros of $U_0$ in $[0,1]$ can be made arbitrarily large with the appropriate choice of $\alpha_0$.\\[-11pt]

\noindent In the interval $-1 -\pi^2 < \alpha_0 < -1$ the background current $U_0$ has one zero at $Y_0 = 1 - \lambda/\theta_0$ for $\arccot(\theta_1 \cot(\theta_1)/\theta_0) < \lambda \leq \theta_0$, and none for $\lambda \in (\theta_0,\pi)$.\\[-10pt] 
\item[W~\ref{class:2}.]  $U_0$ has one zero at $Y_0 = 1 - \lambda$ for $\lambda \in (\pi/4,1]$ and none for $\lambda \in (1,\pi)$.\\[-10pt] 
\item[W~\ref{class:3}.] $U_0$ has one zero at $Y_0 = 1 - \lambda/\theta_0$ for $\arccot(\theta_1 \coth(\theta_1)/\theta_0) < \lambda \leq \theta_0$ and none for $\lambda \in (\theta_0,\pi)$.\\[-10pt] 
\item[W~\ref{class:4}.] $U_0$ has one zero at $Y_0 = 1 - \theta_0^{-1} \arctanh\left( \frac{\lambda^2 \theta_0}{\lambda^2 \theta_1 \coth(\theta_1)-1} \right)$ for $\lambda^2 \geq (\theta_1 \coth(\theta_1) - \theta_0 \coth(\theta_0))^{-1}$ and none for other $\lambda \neq 0$.\\[-10pt] 
\end{itemize}
\end{thm}

\begin{figure}
\includegraphics[height=0.35\linewidth,width=0.5\linewidth]{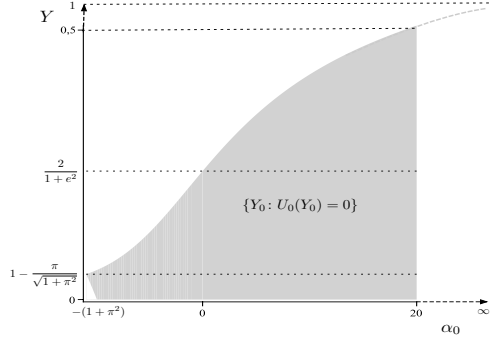}
\caption{\small The possible zeros $Y_0$ of the background current $U_0$ as a (multivalued) function of the laminar vorticity $\alpha_0$; those are the levels of the critical layers in the limit $\varepsilon \to 0$. For each pair $(\alpha_0,Y_0)$ in the shaded region, and only for those, there exists $\lambda$ such that the bifurcation condition is fulfilled and $U_0$ has precisely one zero at $Y_0$. As $\alpha_0 \to \infty$ we have $\max Y_0 \to1$, but for any given $\alpha_0$ the stagnation points are bounded away from the surface as $\varepsilon \to 0$. In contrast, stagnation at the flat bed is possible whenever $\alpha_0 \geq -\pi^2$. For $\alpha_0 < -1-\pi^2$ more zeros appear and the situation is not as transparent.}
\label{fig:zeros}
\end{figure}

\begin{proof}
The analysis is carried out separately for each wave class.

{\bfseries W~\ref{class:1}}. The function
\[
\alpha_0 \mapsto \frac{\theta_1 \cot(\theta_1)}{\theta_0},
\]
spans the real numbers (it blows up at $\alpha_0 = -1 - n^2 \pi^2$, $n \in \N$). We thus see from~\eqref{eq:C1C2<-1} that the set $(\arccot(\theta_1 \cot(\theta_1)/\theta_0), \pi)$ may be empty. But for any $\varepsilon > 0$ and $n \in \N$, there exists $\delta > 0$ such that if $\alpha_0 <  -1 - n^2 \pi^2 < \alpha_0 + \delta$, then \eqref{eq:C1C2<-1} is solvable for all $\lambda \in (\varepsilon,\pi)$. The number of zeros is at least as large as $\lfloor \theta_0/\pi \rfloor \to \infty$ as $\alpha_0 \to -\infty$. The last proposition then follows by checking that $U_0$ can have at most one zero $Y_0$ for $-1 -\pi^2 < \alpha < -1$.

{\bfseries W~\ref{class:2}}.
Consider~\eqref{eq:C1C2-1}. The function
\[
\lambda \mapsto \sin(\lambda) \left( \sin(\lambda) - \cos(\lambda) \right) > 0 \quad\text{ exactly when }\quad \frac{\pi}{4} < \lambda < \pi, 
\]
and since it is bounded, the amplitude $a$ is bounded away from $0$. There thus exist $\lambda$ and $Y_0 \in [0,1]$ such that $U_0(Y_0) = 0$ if and only if $Y_0 = 1 - \lambda  \in [0,1-\textstyle{\frac{\pi}{4}}]$.

{\bfseries W~\ref{class:3}}.
The right-hand side of \eqref{eq:C1C2>-1<0} is positive when
\begin{equation}\label{eq:C2bound}
\arctan \left( \frac{\theta_0}{\theta_1 \coth(\theta_1)} \right) < \lambda < \pi.
\end{equation}
To have $Y_0 \in [0,1]$ with $U_0(Y_0) = 0$ necessarily $\lambda \in [0,\theta_0]$. The assertion follows from that $\theta_0$ is strictly larger than the lower bound in \eqref{eq:C2bound}.
 
{\bfseries W~\ref{class:4}}.
The background current has at most one zero, and to see what zeros there are in $[0,1]$ we consider
\[
\tanh(\theta_0 (Y - 1)) = - \frac{\lambda}{a} = \frac{ \lambda^2 \theta_0 }{1- \lambda^2  \theta_1 \coth(\theta_1) } .
\]
For $Y_0 \in [0,1)$ the left-hand side is negative, so that we must at least have $\lambda^2 > (\theta_1 \coth(\theta_1))^{-1}$, and a closer look yields that $\lambda^2 \geq (\theta_1 \coth(\theta_1) - \theta_0 \coth(\theta_0))^{-1}$ is required to match $Y \geq 0$. The right-hand side is then an increasing function of $\lambda^2$, and
\[
-\tanh(\theta_0) < \frac{ \lambda^2 \theta_0 }{1- \lambda^2  \theta_1 \coth(\theta_1) } \leq \frac{-\theta_0}{\theta_1 \coth(\theta_1)}.
\]   
Since also $\arctanh$ is an increasing function, the question reduces to whether
\[
0 \leq Y_0  <  1- \frac{1}{\theta_0} \, \arctanh\left(\frac{\theta_0 }{\theta_1 \coth(\theta_1) } \right).
\]
The right-hand side is positive, strictly increasing in $\alpha_0$, and tends to $1$ as $\alpha_0 \to \infty$. 
\end{proof}

\section{Hamiltonian formulation and phase-portrait analysis}\label{sec:portraits}
In the analysis to come the region of interest is 
\[
0 \leq Y \leq 1 \pm {\mathcal O}(\varepsilon) \cos(X),
\]   
where the sign indicates that for each of the wave classes~\ref{class:1}--\ref{class:4} one finds that $X=0$ may be either a crest or a trough, depending on the signs of $a$ and $\lambda$. Recall that $(u,v) = (\dot x, \dot y)$. In view of that $X = x- ct$ and $U = u - c$, one similarly obtains 
\[
(\dot X, \dot Y) = (U,V).
\]
The paths $(X(t),Y(t))$ describe the particle trajectories in the steady variables, and any such solution is entirely contained in one streamline. 

\begin{prop}[Hamiltonian formulation]
The wave classes \ref{class:1}--\ref{class:4} all admit a Hamiltonian
\begin{equation}\label{eq:hamiltonian}
H(X,Y) := \varepsilon \cos(X) G(Y) + \int_0^Y U_0(s)\,ds, 
\end{equation}
with
\begin{equation}\label{eq:G}
G(Y) :=
\begin{cases}
\sin(\theta_1 Y), \qquad&\text{ for wave class}~\ref{class:1},\\
Y, \qquad&\text{ for wave class}~\ref{class:2},\\
\sinh(\theta_1 Y), \qquad&\text{ for wave classes}~\ref{class:3}\text{--}\ref{class:4}.\\ 
\end{cases}
\end{equation}
\end{prop}
\medskip
Classes \ref{class:2}--\ref{class:4} can be dealt with as the class $\alpha = 0$ (constant vorticity) in \cite{MR2409513,Wahlen09}, and do not yield any new qualitative results. Indeed, their appearance and the analysis thereof is captured within that of the interesting class~\ref{class:1}.

\begin{figure}
\includegraphics[width=0.4\linewidth]{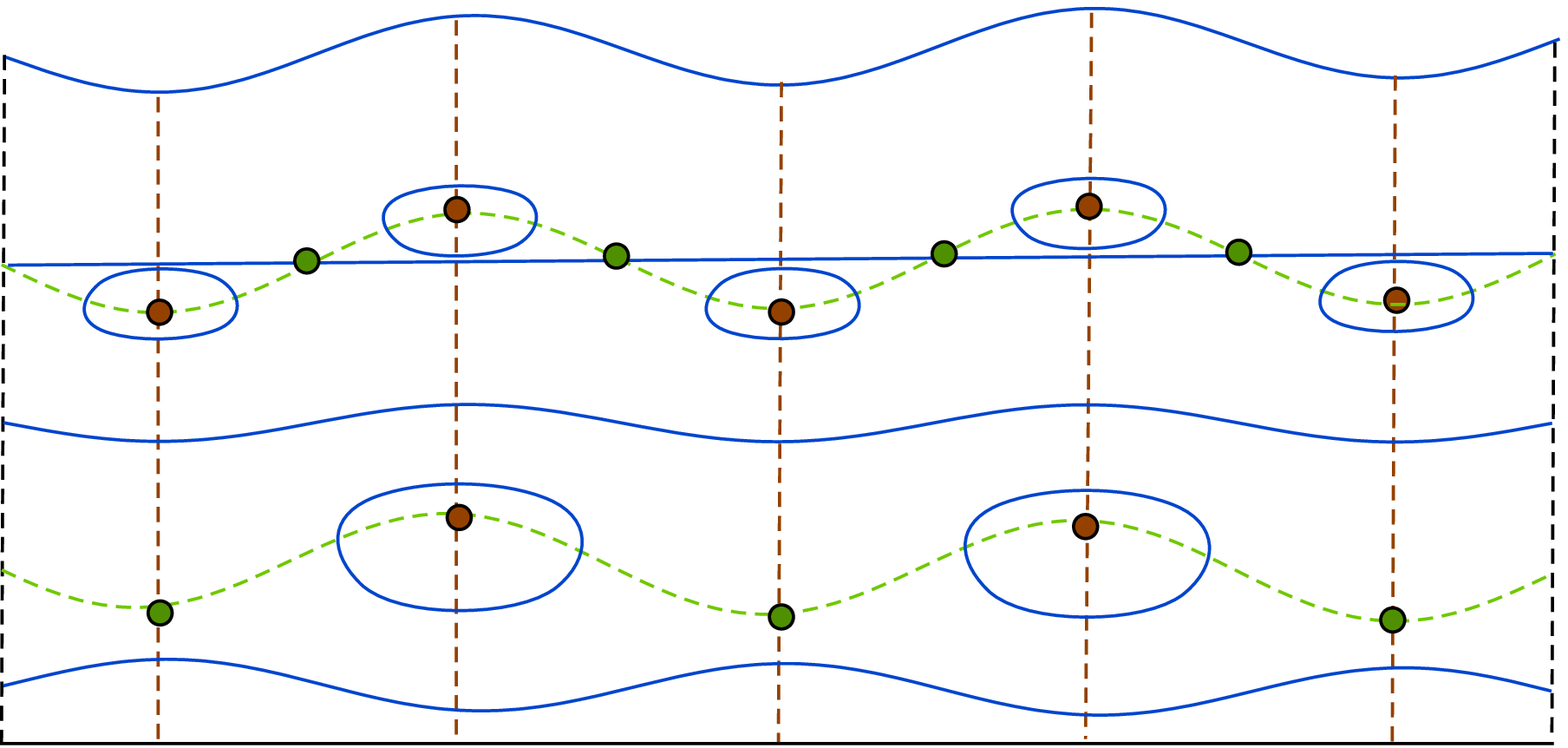}\qquad
\includegraphics[width=0.4\linewidth]{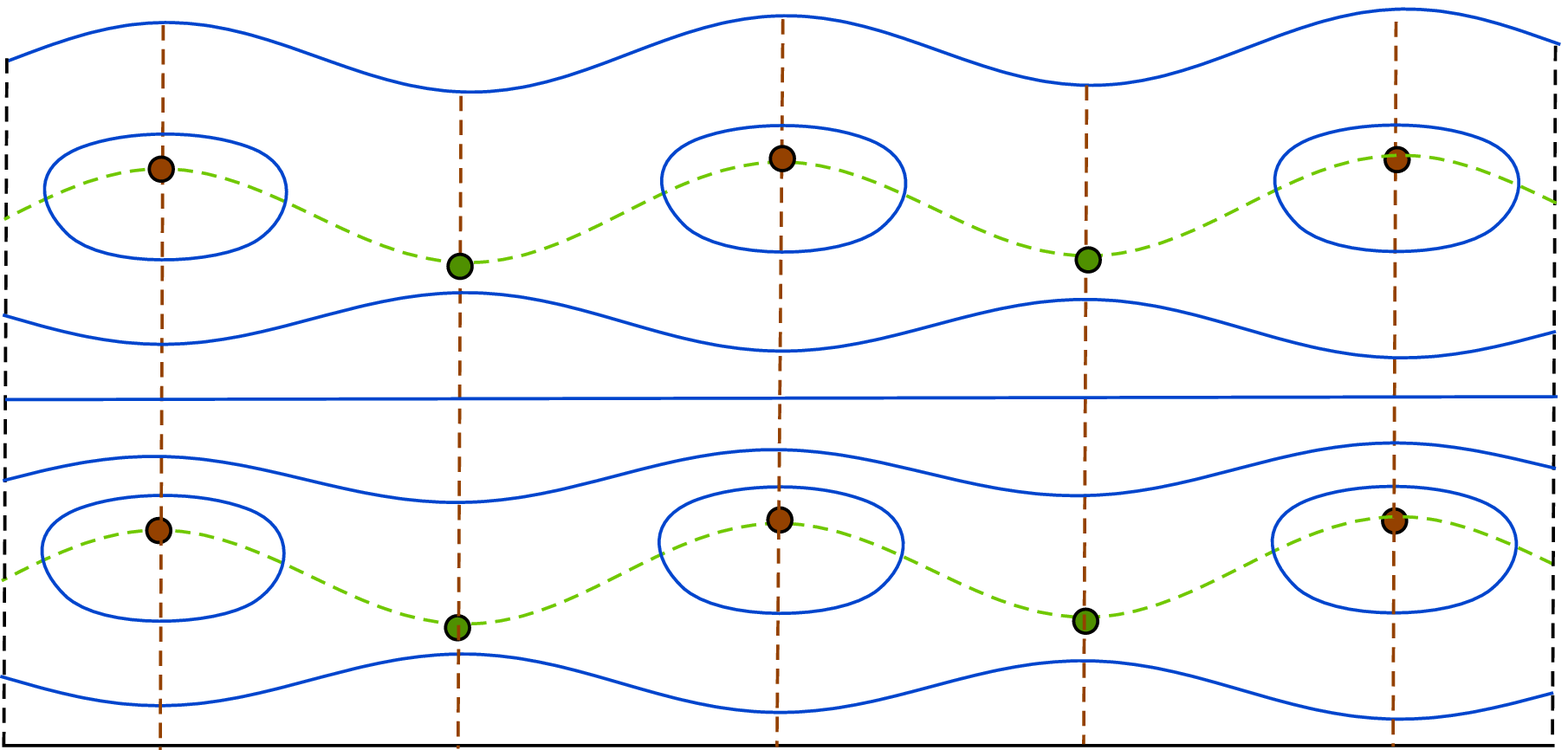}
\caption{\small Two scenarios for wave class~\ref{class:1}. Blue lines are streamlines, dotted lines isoclines [$\infty$-isocline green, $0$-isocline red], and fat dots critical points [centers red, saddle points green]. Left: an uppermost critical layer with horizontal streamlines cutting through it as in Thm.~\ref{thm:main}~ii.b), and a lower critical layer as in Thm.~\ref{thm:main}~ii.a). Right: two critical layers as in Thm.~\ref{thm:main}~ii.a) separated by a horizontal streamline. Note that the rotational flow near the bottom evolves under a ``rigid lid''.  }
\label{fig:wave24}
\end{figure}

\begin{thm}[Wave class 1]
\label{thm:main}
The following hold for small-amplitude waves of wave class~\ref{class:1} ($\varepsilon$ sufficiently small).
\begin{itemize}
\item[i.] The fluid motion is divided into vertical layers, each separated from the others by flat sets of streamlines $\{(X,Y_*) \colon \sin(\theta_1 Y_*) = 0\}$.\\[-10pt]
\item[ii.] For each $Y_*$ with $U_0(Y_*) =0$ there is a smooth connected part of the $\infty$-isocline passing through all points $(\pi/2 + n\pi, Y_*)$, $n \in \Z$,  along which centers (cats-eye vortices) and saddle points alternate in one of the following ways: 
\begin{itemize}
\item[a)] when $Y_*$ is not a common zero of $U_0$ and $\sin(\theta_1 \cdot)$ centers appear at every other $X = n\pi$ and saddle points at every other $(n+1)\pi$;\\[-10pt]    
\item[b)] when $Y_*$ is a common zero of $U_0$ and $\sin(\theta_1 \cdot)$ centers appear at $X = n\pi$ and saddle points at $\pi/2 + n\pi$.\\[-10pt]    
\end{itemize}
\end{itemize}
\end{thm}

\begin{remark}\label{rem:main2}
Starting with the situation in ii.b) one might fix $\varepsilon$ and $\alpha$, and then vary the zero of the background flow. The saddle point at $X = \pi/2$ then continuously and monotonically approaches the center at either $X=0$ or $X=\pi$, eventually merging with it and wiping it out.    
\end{remark}



\begin{figure}
\includegraphics[height=0.3\linewidth,width=0.4\linewidth]{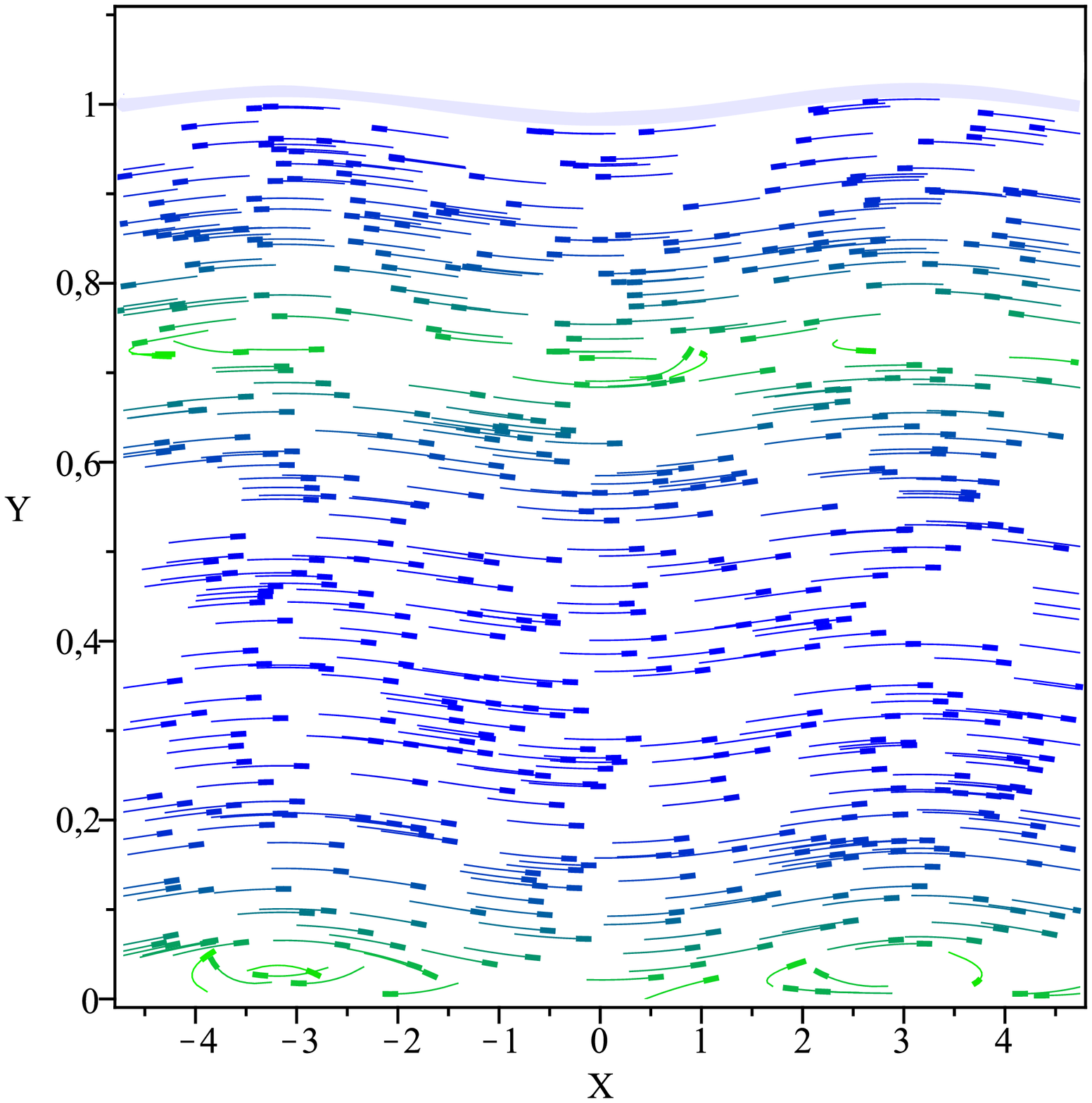}\qquad
\includegraphics[height=0.3\linewidth,width=0.4\linewidth]{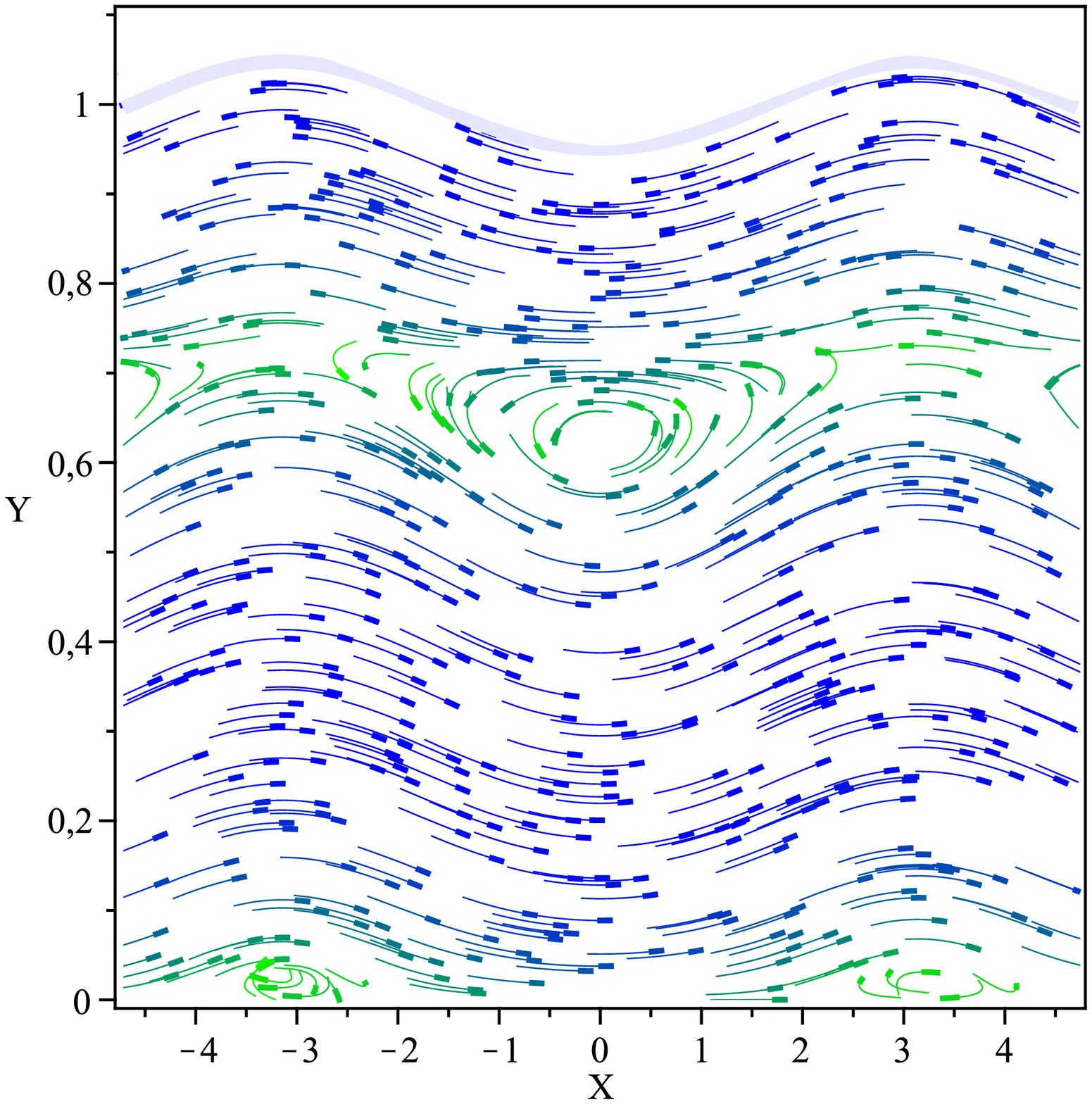}
\includegraphics[height=0.3\linewidth,width=0.4\linewidth]{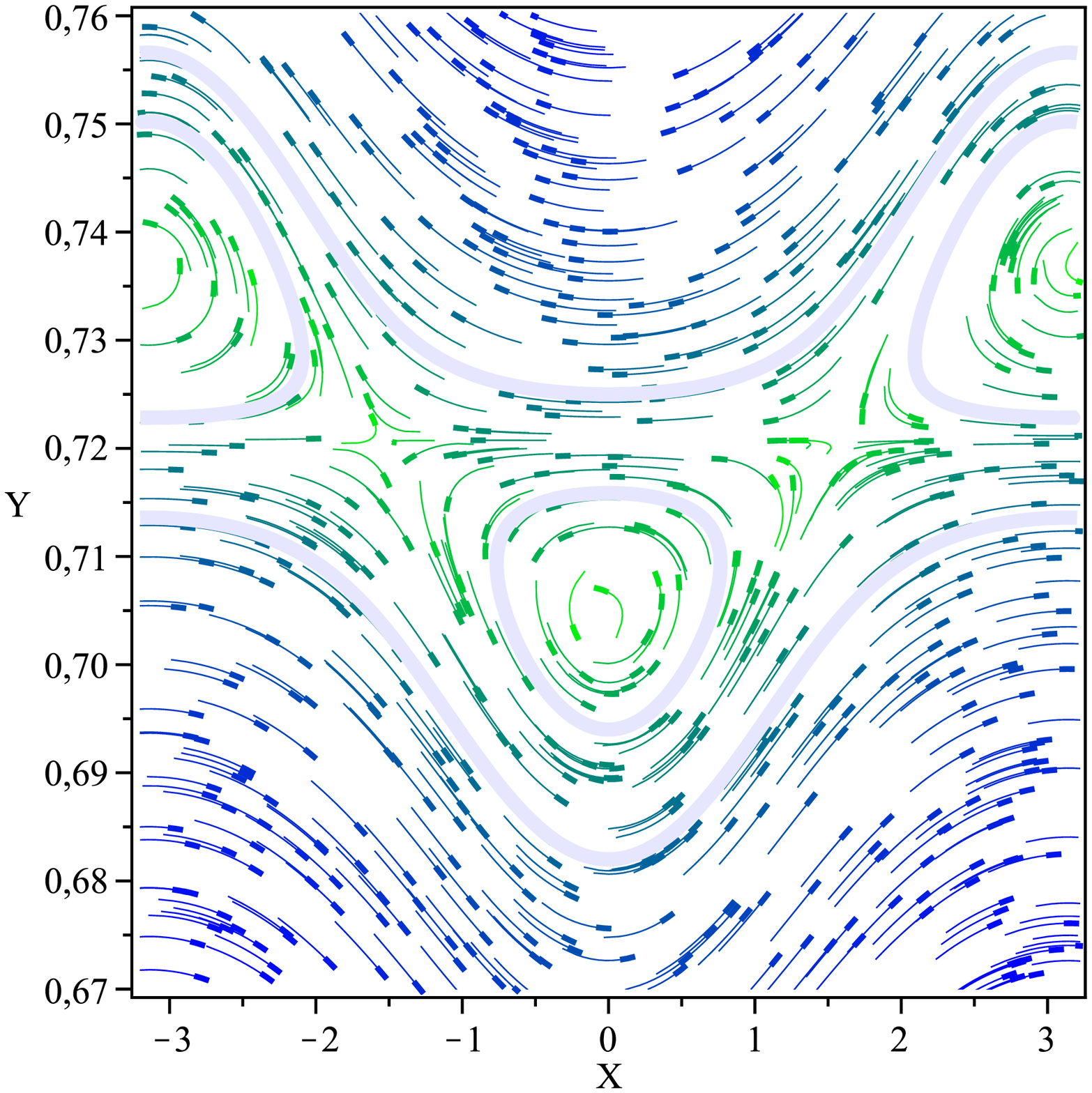}\qquad
\includegraphics[height=0.3\linewidth,width=0.4\linewidth]{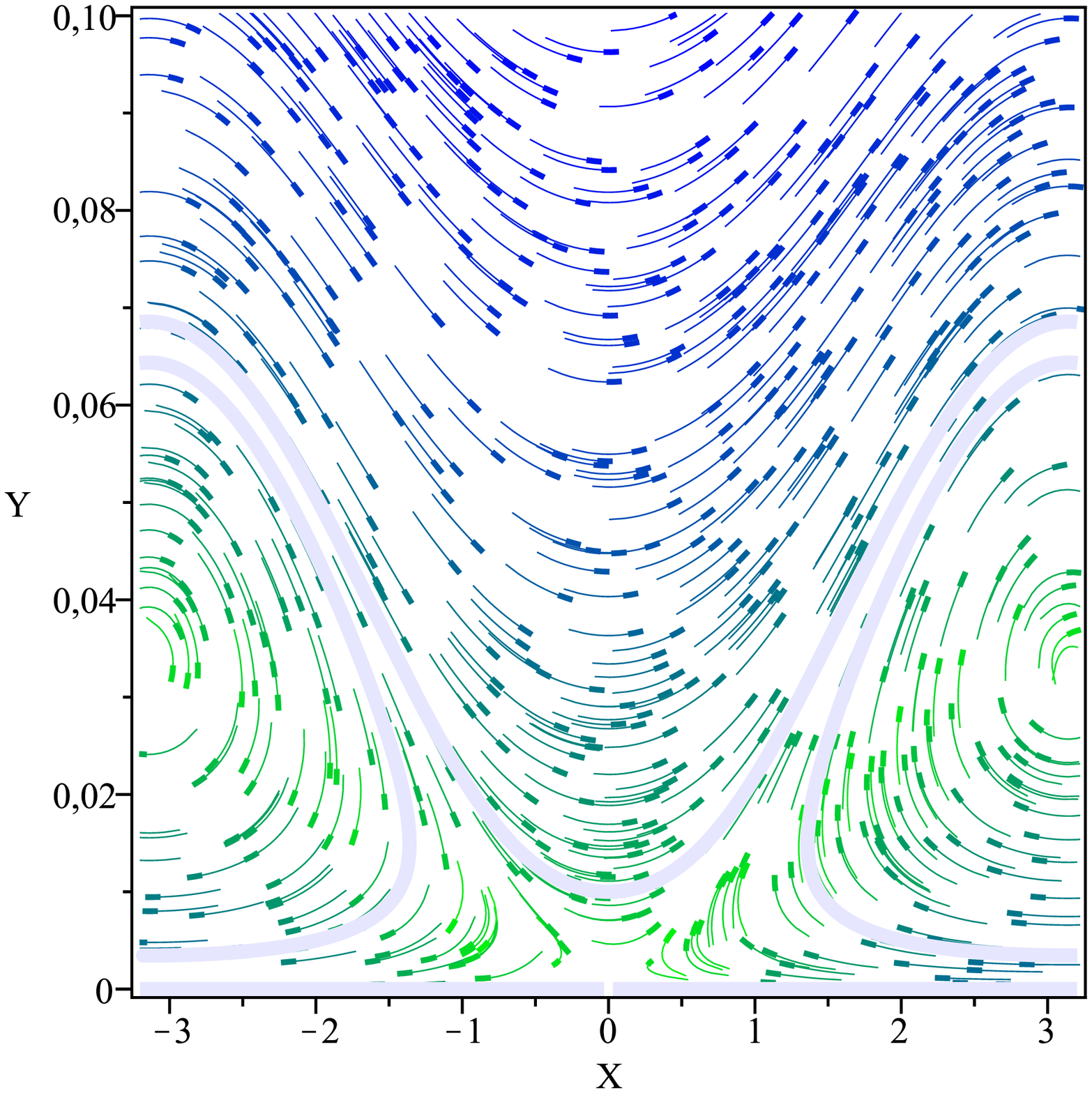}
\caption{\small Numerical plots of wave class~\ref{class:1}. Colours indicate the strength of the velocity field. Top left [$\alpha=-20$, $\lambda=4.39$,  $\varepsilon = 0.05$]: the case of a common zero as in Thm.~\ref{thm:main}~ii.b). The upper critical layer is magnified in the plot bottom left, where the three centers, the two saddle points and the $0$-isocline cutting through them are clearly visible. The bottom right plot shows the same flow near the bottom with only two centers as in Thm.~\ref{thm:main}~ii.a). The situation is similar to that in Figure~\ref{fig:wave24}, left, with the difference that the $\infty$-isocline actually crosses the flat bed in this numerical plot. Right [$\alpha=-20$, $\lambda=4.60$,  $\varepsilon = 0.05$]: as the zero of the background current shifts, the horizontal $0$-isocline climbs above the critical points in the upper critical layer and the saddle nodes merge with one center each, turning the situation in Thm.~\ref{thm:main}~ii.b) to the one in Thm.~\ref{thm:main}~ii.a).}
\label{fig:plot14}
\end{figure}

\begin{proof}
It follows from \eqref{eq:hamiltonian} that the fluid motion is $2\pi$-periodic and symmetric around the vertical $X=0$ axis. It therefore suffices to investigate the strip $0 \leq X \leq \pi$.

\subsection*{i)} The velocity field is 
\begin{align*}
\dot X &= U_0(Y) + \varepsilon \theta_1 \cos(X) \cos(\theta_1 Y),\\ 
\dot Y &= \varepsilon \sin(X) \sin(\theta_1 Y),
\end{align*}
whence the \emph{$0$-isocline}, defined as the set where $\dot Y=0$, is given by the vertical axes $X = 0$ mod $\pi$ and the horizontal lines where $\sin(\theta_1 Y) = 0$. 

\subsection*{ii)} Let $Y_*$ be a zero of $U_0$. Then $(\pi/2,Y_*)$ belongs to the \emph{$\infty$-isocline} $\{(X,Y) \colon \dot X = 0\}$. Since $U_0^\prime(Y_*) = a \theta_0 \neq 0$, we have that
\[
D_Y \dot X = U_0^\prime(Y) - \varepsilon \theta_1^2 \cos(X) \sin(\theta_1 Y)
\]
is nonzero at $(\pi/2,Y_*)$. The implicit function theorem allows us to locally parametrize the $\infty$-isocline as the graph of a smooth function $Y_{\infty}(X)$ with slope 
\begin{equation}\label{eq:slope}
\Diff_X Y_\infty = \frac{\varepsilon \theta_1 \sin(X) \cos(\theta_1 Y)}{U_0^\prime(Y) - \varepsilon \theta_1^2 \cos(X) \sin(\theta_1 Y)}.
\end{equation}
A continuity argument yields that, for $\varepsilon$ small enough, $Y_\infty$ extends to a $2\pi$ periodic function on $\R$, strictly rising and falling between the zeros of $\sin(X)$, and with $|Y_\infty - Y_*| = \OO(\varepsilon)$. Hence $U_0^\prime(Y_\infty) \neq 0$. At $X = n\pi$, $n \in \Z$, the graph of $Y_\infty$ intersects the $0$-isocline. At those \emph{critical points}, the Hessian of the Hamiltonian is given by
\begin{equation}\label{eq:hessian0}
\begin{aligned}
&\Diff^2 H(n\pi,{Y_\infty}|_{X=n\pi})\\ 
&= 
\begin{bmatrix}
(-1)^{n+1}\varepsilon  \sin(\theta_1 {Y_\infty}) &  0\\
0 & (-1)^{n+1} \varepsilon \theta_1^2 \sin(\theta_1 {Y_\infty}) + U_0^\prime({Y_\infty})
\end{bmatrix}\bigg|_{X= n\pi},
\end{aligned}
\end{equation}      
There are now two cases.

\subsection*{a) When $\sin(\theta_1 \cdot)$ and $U_0$ have no common zero}
From $|Y_\infty - Y_*| = \OO(\varepsilon)$ we find that $\sin(\theta_1 Y_\infty)$ is non-vanishing and thus of constant sign. For $\varepsilon$ small enough, the  Hessian~\eqref{eq:hessian0} thus has one negative and one positive eigenvalue at every other $X = n\pi$, and two of the same sign at every other $X = (n+1) \pi$ in between. The assertion~ii.a) then follows from the Morse lemma \cite{MR0163331}.

\subsection*{b) When $\sin(\theta_1 \cdot)$ and $U_0$ have a common zero}
In this case there are additional critical points at $X = \pi/2 + n\pi$, $n \in \Z$, all similar to the one at $X = \pi/2$. There
\[
\Diff^2 H(\pi/2,Y_*) = 
\begin{bmatrix}
0 & -\varepsilon \theta_1 \cos(\theta_1 Y_*) \\
-\varepsilon \theta_1 \cos(\theta_1 Y_*) &  U_0^\prime(Y_*)
\end{bmatrix},
\]
with one strictly positive and one strictly negative eigenvalue. Hence, the critical point $(\pi/2, Y_*)$ is always a saddle point. 

We want to show that the Hessian \eqref{eq:hessian0} has two eigenvalues of the same sign at all critical points $(n\pi, Y_\infty|_{X = n\pi})$. Since the slope of $Y_\infty$ changes direction exactly at $X = n\pi$ it follows that also in the case when $U_0(Y_*) = \sin(\theta_1 Y_*) = 0$ we have $\sin(\theta_1 Y_\infty|_{x = n\pi}) \neq 0$, but with
\[
\sign \sin(\theta_1 Y_\infty|_{X = n\pi} ) = -\sign \sin(\theta_1 Y_\infty|_{X = n\pi} ),
\]
all given that $\varepsilon$ is small enough. We now claim that $-\varepsilon  \sin(\theta_1 {Y_\infty}|_{X=0})$ and $U_0^\prime(Y_\infty|_{X=0})$ have the same sign (cf. \eqref{eq:hessian0}). The slope of $Y_\infty|_{X\in(0,\pi)}$ is determined by the sign of $\cos(\theta_1 Y_\infty)/U_0^\prime(Y_\infty)$ in the same interval. We have $\cos(\theta_1 \cdot) > 0$ when $\sin(\theta_1 \cdot)$ is increasing, and contrariwise. The assertion now follows from that $\sin(\theta_1 Y_*) = 0$ .
\end{proof}

\subsection*{Acknowledgment} ME would like to thank Erik Wahl{\'e}n for interesting discussions. Part of this research was carried out during the Program ``Recent advances in integrable systems of hydrodynamic type'' at the Erwin Schr\"odinger Institute, Vienna; ME and GV thank the organizers for the nice invitation. The authors also acknowledge the comments and suggestions made by the referee, which helped develop the exposition.

\bibliographystyle{siam}
\bibliography{critical_linear}

\end{document}